\documentclass{birkjour}

\usepackage{amssymb}
\usepackage{amsmath}
\usepackage{hyperref}
\newcommand{\BR}{{\mathbb R}}

\newcommand{\BZ}{{\mathbb Z}}
\newcommand{\BC}{{\mathbb C}}
\newcommand{\e}{{\bf e}}
\renewcommand{\a}{{\bf a}}
\renewcommand{\b}{{\bf b}}
\newcommand{\f}{{\bf f}}
\newcommand{\g}{{\bf g}}
\renewcommand{\u}{{\bf u}}
\renewcommand{\v}{{\bf v}}
\newcommand{\Th}{{\bf T}_h} 

\newcommand{\cl}{C \kern -0.1em \ell}

\newtheorem{remark}{Remark}[section]
\newtheorem{lemma}{Lemma}[section]
\newtheorem{proposition}{Proposition}[section]
\newtheorem{corollary}{Corollary}[section]

\begin{document}

%
%
%
%
%
%
%
%
%

\title[Solutions for the Klein-Gordon and Dirac equations on the
 lattice]
 {Solutions for the Klein-Gordon and Dirac equations on the
 lattice based on Chebyshev polynomials}

\author[N.~Faustino]{Nelson Faustino}

\address{%
Departamento de Matem\'atica Aplicada \\ IMECC--Unicamp\\
CEP 13083--859, Campinas--SP, Brasil}

\email{\href{mailto:faustino@ime.unicamp.br}{faustino@ime.unicamp.br}}

\thanks{Research supported by the fellowship
\href{http://www.bv.fapesp.br/14523}{13/07590-8} of FAPESP (S.P., Brazil).}

\subjclass{Primary 30G35, 39A12; Secondary 33C05, 53Z05}

\keywords{Chebyshev polynomials; discrete Dirac operators; lattice
fermion doubling; spinor fields.}

\date{\today}

\begin{abstract}
The main goal of this paper is to adopt a multivector calculus
scheme to study finite difference discretizations of Klein-Gordon
and Dirac equations for which Chebyshev polynomials of the first
kind may be used to represent a set of solutions. The development of
a well-adapted discrete Clifford calculus framework based on spinor
fields allows us to represent, using solely projection based
arguments, the solutions for the discretized Dirac equations from
the knowledge of the solutions of the discretized Klein-Gordon
equation. Implications of those findings on the interpretation of
the lattice fermion doubling problem is briefly discussed.
\end{abstract}

\maketitle

\section{Introduction}

\subsection{State of art}\label{StateArtSubsection}

For a variety of reasons, the study of equations from relativistic
wave mechanics through the incorporation of a fermionic lattice
structure in the discrete space-time, plays an important role far
beyond the design of non-perturbative methods in Quantum
Electrodynamics (QED) and Quantum Chromodynamics (QCD)
(cf.~\cite[Chapters 4 \& 5 ]{MontvayMunster94}). Such kind of
lattice structure, also used in the study Higgs and Yukawa models
(cf. \cite[Chapter 6]{MontvayMunster94}), was widely popularized
during the last decade from its crucial role on the representation
of Ising models as {\it local} conformal structures of spin-type
(cf.~\cite{Mercat01,ChelkakSmirnov12}).

Historically it was D.~Bohm (cf. \cite{Bohm65}) one of the firsts
that recognizes such need and E.A.B.~Cole \cite{Cole70} one amongst
many that present the former routes addressing to this topic. With
the seminal works of K.~Wilson \cite{Wilson74} and Kogut-Susskind
\cite{KS75}, involving lattice regularizations of Dirac equations,
it was realized that there is a spectrum degeneracy phenomenon
provided by the replication of fermionic states in the massless
limit $m \rightarrow 0$ for the resulting discretized equations, the
so well-know {\it lattice fermion doubling} problem.

Years later, J.~M.~Rabin (cf.~\cite{Rabin82}) explained with some
detail that such gap is indeed a direct consequence of the
non-trivial topology of the lattice momentum space supplied by a
cut-off. Such kind of topology, isomorphic to a $n-$dimensional
torus $(\BR/h\BZ)^n \cong \BR^n/h\BZ^n$, corresponds in the momentum
space to the restriction of all momenta to the cube
$Q_h=\left[-\frac{\pi}{h},\frac{\pi}{h}\right]^n$, the so-called
{\it Brillouin zone} (cf. \cite{Froyen89}). Moreover, it was
explained in detail that this gap is only filled by the chiral
breaking of the symmetries provided by the {\it continuum} model.
Further details arising this kind of construction may also be found
on the papers \cite{Borici08,Creutz08}.

The proper mathematical foundation of {\it lattice fermion
doubling}, traced in depth by Nielsen--Ninomiya in \cite{NN81} and
proved afterwards by D.~Friedan in \cite{Friedan82}, shows that
through lattice regularizations of fermionic fields the existence of
doublers always hold when one impose translational invariance,
hermiticity and locality constraints. That is, the doubling of
solutions underlying discretized models always arise under this set
of constraints (cf.~\cite[Sections 2 \& 3]{Friedan82}).

Driven by the combination of these ideas with Becher-Joos's insights
\cite{BecherJoos82} on Dirac-K\"ahler fermions, other kinds of
lattice formulations such as the approaches obtained by J.~Vaz
\cite{Vaz97} and Kanamori--Kawamoto \cite{KK04} were obtained. The
methods employed on both formulations were essentially build up from
Dimakis--M\"uller-Hoissen approach on
  noncommutative differential calculus over discrete sets (cf. \cite{DimakisHoissen94}).

In a different context, mainly driven by the need of obtaining
factorizations for discrete Laplacians on combinatorial surfaces,
I've introduced in collaboration with U.~K\"ahler and F.~Sommen
(cf.~\cite{FKS07}) finite difference approximations for the Dirac
operator in {\it continuum}. In the spirit of multivector calculus,
it was explained on my PhD~dissertation~\cite{Faustino09} that such
kind of discretization is interrelated with the Dirac-K\"ahler
formalism considered in \cite{BecherJoos82,Vaz97,KK04}.

Interestingly enough (but not yet fully adopted or known in the
community) combination of tools from finite difference potentials
(cf.~\cite{GH01,CKKS14}) and interpolation theory
(cf.~\cite{Froyen89,ShaFannPaulino03}) arising in this context may
also be useful in the modelling of problems of quantum field theory
over the phase space $h\BZ^n \times
\left[-\frac{\pi}{h},\frac{\pi}{h}\right]^n$ (see, for instance,
\cite{MonacoCapelas94,daVeiga02}).

\subsection{Outline of the paper}

In this paper we will show the feasibility of discrete Clifford
calculus in the exact representation of the solutions of some
discretized equations from wave mechanics. To this end, a consistent
multivector calculus scheme through the lattice $h\BZ^n$ will be
introduced in Section \ref{MultivectorCalculusLattice} with the aim
of investigate in Section \ref{KleinGordonDiracSection} the
solutions of the discretized time-harmonic Klein-Gordon equation
\begin{eqnarray}
\label{KleinGordonLattice}\Delta_h \f(x)=m^2\f(x)
\end{eqnarray}
for a given finite difference approximation $\Delta_h$ of the
Laplace operator $\Delta=\sum_{j=1}^n\partial_{x_j}^2$, and
moreover, the solutions of a discretized Dirac equation from the
knowledge of the solutions of (\ref{KleinGordonLattice}). Such
characterization corresponds in the paper to Proposition
\ref{DiracLatticeProposition} and Corollary
\ref{DiracLatticeCorollary}.

The problems of foremost interest treated in Section
\ref{KleinGordonDiracSection} also involve the hypercomplex
extension of the Chebyshev polynomials of the first kind
\begin{eqnarray}
\label{ChebyshevPoly} T_{k}(\lambda)=\frac{1}{2}\left(
\lambda+\sqrt{\lambda^2-1}\right)^k+\frac{1}{2}\left(\lambda-\sqrt{\lambda^2-1}\right)^k
\end{eqnarray}
underlying to the hypergeometric series representation
${~}_2F_1\left( -k,k;\frac{1}{2};\frac{1-\lambda}{2} \right)$.

The choice of this kind of polynomials, also ubiquitous in E.A.B.
Cole's former approach (cf. \cite[p.~650]{Cole70}), was motivated
from its wide range of applications far beyond the computation of
  discrete cosine transforms (cf. \cite{ShaFannPaulino03}).

Following the same train of thought of Bor$\check{s}$tnik-Nielsen's
seminal paper on multivector calculus (cf.~\cite{BorsNielsen00}), we
will look further in Section \ref{MultivectorCalculusLattice} to the
emboid of a Clifford algebra with signature $(0,n)$ onto the algebra
$\mathcal{A}_h$ of all real-valued lattice functions in $h\BZ^n$.

To this end, one will formulate the wedge ($\wedge$) and the dot
($\bullet$) product between a Clifford generator and a multivector
function on the lattice by taking into account two kinds of
displacement actions on the axis of $\BR^n$ in the same flavor of
noncommutative differential calculus carrying a discrete set
(cf.~\cite{DimakisHoissen94}). In brief,~an {\it involution} action
combined with forward and/or backward shifts on $h\BZ^n$ will be
adopted with the aim of filling the noncommutative gap provided by
the lack of a {\it true} Leibniz rule for the standard finite
difference operators. As it will be shown in Lemma
\ref{nilpotencyWittDh}, such kind of constraints assures a
generalized Leibniz rule for the corresponding finite difference
operators of multivector type.

Based on the aforementioned facts, the resulting discrete Clifford
calculus -- i.e. a geometric based extension of discrete vector
calculus-- obtained through the Fock space representation of
$\mbox{End}(\Lambda^*\mathcal{A}_h)$ (cf.~\cite[Section
3.2]{Faustino09}) will allows us to describe the basic left/right
endomorphisms acting on $\Lambda^*\mathcal{A}_h$ as the canonical
equivalents of the generators of a Clifford algebra with signature
$(n,n)$. In opposition to \cite{Vaz97,KK04}, on which the resulting
Clifford geometric products are distributive but non-associative,
the Fock space representation encoded by the wedge and dot products
avoid {\it a priori} problems related with associativity and
distributivity on the product between a Clifford basis and a
multivector function with membership in $\Lambda^*\mathcal{A}_h$.

The major challenging here against \cite{BorsNielsen00} will be the
introduction of a unitary action over the resulting multivector
space through a {\it local} action $\chi_h(x)$ on $h\BZ^n$. Such
action is closely related with a pseudoscalar representation within
the Clifford algebra of signature $(n,n)$. As one will see in
Section \ref{FermionDoublingSection}, such {\it local} action of
unitary type incorporates the Kogut-Susskind fermion regularization
(see also \cite{Borici08,Creutz08} for further comparisons).

The intriguing aspect besides this staggered based formulation is
that such lattice actions will also endow projection operators that
provide, as in \cite{RochaVaz06}, a direct sum decomposition
involving exterior algebras of chiral and achiral type, similar to
spinor spaces $\frac{1}{2}\left(1\pm\gamma\right)\cl_{0,n}$ that
appear on the direct sum decomposition
$$
\cl_{n,n}=\frac{1}{2}\left(1+\gamma\right)\cl_{0,n}
\oplus\frac{1}{2}\left(1-\gamma\right)\cl_{0,n}
$$
so that $\gamma$ is a pseudoscalar of $\cl_{n,n}$ satisfying
$\gamma^2=1$ (cf.~\cite[Sections IV. \& V.]{BorsNielsen00}).

\section{Multivector calculus on the lattice}\label{MultivectorCalculusLattice}

\subsection{Real Clifford algebras}

Let $\BR^{p+q}$ be the standard $(p+q)-$dimensional Euclidean space.
The Clifford algebra over $\BR^{p+q}$ with signature $(q,p)$
corresponds to a real associative algebra with identity $1$,
containing $\BR$ and $\BR^{p+q}$ as subspaces, and in which the
basis elements $\e_1,\e_2,\ldots,\e_{p},\e_{p+1},\ldots,\e_{p+q}$
satisfy the following graded anti-commuting relations
\begin{eqnarray}
\label{CliffordBasis}
\begin{array}{lll}
\e_j \e_k+ \e_k \e_j=-2\delta_{jk}, & j,k=1,2,\ldots,p \\
\e_{j} \e_{k+q}+ \e_{k+q} \e_{j}=0, & j,k=1,2,\ldots,p \\
\e_{j+q} \e_{k+q}+ \e_{k+q} \e_{j+q}=2\delta_{jk}, &
j,k=1,2,\ldots,p.
\end{array}
\end{eqnarray}

Based on the above set of constraints, a basis $\e_J$ for
$\cl_{q,p}$ then consists of $\e_{J}=\e_{j_1}\e_{j_2}\ldots
\e_{j_r}$, where $J=\{j_1,j_2,\ldots,j_r\}$ is a partially ordered
subset of $\{ 1,2,\ldots,p,p+1,\ldots,p+q\}$ with cardinality
$|J|=r$ so that $0 \leq r\leq p+q$. For $J=\varnothing$ (empty set)
we will use the convention $\e_\varnothing=1$.

Then, any Clifford number $\a\in \cl_{q,p}$ may thus be expressed as
 \begin{eqnarray*}
\a=\sum_{r=0}^{p+q}\sum_{|J|=r} a_J ~\e_J.
\end{eqnarray*}

 Herewith, it is important to notice that $\cl_{q,p}$ is a
universal algebra of dimension $2^{p+q}$, linear isomorphic to the
exterior algebra $\Lambda^*\left(\BR^{p+q}\right)$ (cf.
\cite[Chapter 2]{RodriguesOliveira07}). In such way, the elements of
$\BR$ in $\cl_{q,p}$ are represented as $\a=a~\e_\varnothing$
whereas the vectors of $\BR^{p+q}$ correspond in $\cl_{q,p}$ to the
ansatz $\displaystyle z=\sum_{j=1}^p x_j\e_{j}+y_{j}\e_{j+q}$.
Notice also that $x=\sum_{j=1}^p x_j\e_j$ corresponds to the
$\cl_{0,p}-$valued representation of the $p-$ tuple of
$(x_1,x_2,\ldots, x_p)$ of $\BR^p$ whereas $\displaystyle
y=\sum_{j=1}^p y_{j} \e_{j+q}$ corresponds to $\cl_{q,0}-$valued
representation of the $q-$tuple $(y_{1},y_{2},\ldots, y_{q})$ of
$\BR^q$.

There are essentially three automorphisms that leave the multivector
structure of $\cl_{q,p}$ invariant. They are defined as follows:
\begin{itemize}
\item The {\it main involution} $\a \mapsto \a'$ is defined
recursively via
\begin{eqnarray}
\label{involution}
\begin{array}{lll}
(\a \b)'=\a'\b' \\ (a_J \e_J)'
=a_J~\e_{j_1}'\e_{j_2}'\ldots \e_{j_r}' \\
\e_j'=-\e_j ~~~\mbox{and}~~~ \e_{j+q}'=\e_{j+q}.
\end{array}
\end{eqnarray}
\item The {\it reversion} $\a \mapsto \a^*$ is defined recursively
via
\begin{eqnarray}
\label{reversion}
\begin{array}{lll}
(\a \b)^*=\b^*\a^* \\ (a_J \e_J)^*
=a_J~\e_{j_r}^*\ldots \e_{j_2}^*\e_{j_1}^*~~~ \\
\e_j^*=\e_j~~~\mbox{and}~~~\e_{j+q}^*=\e_{j+q}.
\end{array}
\end{eqnarray}
\item The $\dag-${\it conjugation} $\a \mapsto \a^\dag$ is defined
recursively via
\begin{eqnarray}
\label{conjugation}
\begin{array}{lll}
(\a \b)^\dag=\b^\dag\a^\dag \\ (a_J \e_J)^\dag =a_J~\e_{j_r}^\dag
\ldots \e_{j_2}^\dag\e_{j_1}^\dag \\
\e_j^\dag=-\e_j~~~\mbox{and}~~~\e_{j+q}^\dag=\e_{j+q}.
\end{array}
\end{eqnarray}
\end{itemize}

Notice that the $\dag-${\it conjugation} may be rewritten as a
composition between {\it main involution} and {\it reversion}, that
is $(\a')^*=(\a^*)'=\a^\dag$ holds for every $\a\in \cl_{q,p}$
whereas in case when $\a$ belongs to $\cl_{0,p}$ the Clifford number
$\a^\dag$ coincides with the standard conjugation $\overline{\a}$
over $\cl_{0,p}$ i.e. $\a^\dag=\overline{\a}$. So, any Clifford
vector $x=\sum_{j=1}^p x_j \e_j$ of $\cl_{0,p}$ satisfies $x^*=x$
and $x^\dag=x'=-x$.

Let us now turn our attention to the description of the real
Clifford algebra $\cl_{n,n}$ ($q=p=n$) as a canonical realization of
the algebra of endomorphisms $\mbox{End}(\cl_{0,n})$.

 First,
recall that in terms of the {\it main involution} automorphism, the
Clifford product $\e_j\a$ may be decomposed as
$$
\e_j \a=\frac{1}{2}\left( \e_j \a+(\a\e_j)'\right)+\frac{1}{2}\left(
\e_j \a-(\a\e_j)'\right).
$$

The symmetric and skew-symmetric parts of the above summand thus
give rise to the multivector counterparts for the dot ($\bullet$)
and wedge ($\wedge$) product. Indeed, introducing $\e_j \bullet \a$
and $\e_j \wedge \a$ as
\begin{center}
$\e_j \bullet \a=-\frac{1}{2}\left( \e_j \a-\a'\e_j\right)$ and
$\e_j \wedge \a=\frac{1}{2}\left( \e_j \a+\a'\e_j\right)$
\end{center}
it follows then $\e_j\a= -\e_j \bullet \a+ \e_j \wedge \a$.

Next, let us identify the generators of $\cl_{n,n}$ as left and
right endomorphisms via the set of canonical correspondences
\begin{eqnarray}
\label{EndCl0n}\e_{j}: \a \mapsto \e_j\a&~~\mbox{and}~~ & \e_{j+n}:
\a \mapsto \a'\e_j.
\end{eqnarray}

From the combination of (\ref{EndCl0n}) with (\ref{CliffordBasis})
it is clear that the set of relations $\e_j(\e_j\a)=\e_j^2\a=-\a$
and $(\a'\e_j)'\e_j=\e_{j+n}^2\a=\a$ reveal the canonical
isomorphism between $\mbox{End}(\cl_{0,n})$ and $\cl_{n,n}$ provided
by (\ref{EndCl0n}).

  The dot and wedge products defined above suggest the introduction
of the set of operators $\e_j^+:\a \mapsto \e_j \bullet \a$ and
$\e_j^-:\a\mapsto \e_j \wedge \a$, the so-called Witt basis. From
(\ref{EndCl0n}) $\e_j^+=\e_j \bullet (\cdot)$ and $\e_j^-=\e_j
\wedge ( \cdot)$ may be rewritten as
\begin{eqnarray}
\label{WittBasis} \e_j^+=\frac{1}{2}\left( \e_{j+n}-\e_j\right)
&~~\mbox{and}~~& \e_j^-=\frac{1}{2}\left( \e_{j+n}+\e_j\right).
\end{eqnarray}

Therefore, the set $\{\e_j^+,\e_j^-~:~j=1,2,\ldots,n\}$ also forms a
basis for $\mbox{End}(\cl_{0,n})$. The remaining set of graded
anti-commuting relations are given by
\begin{eqnarray}
\label{WittBasisjk}
\begin{array}{lll}
\e_j^-\e_k^- + \e_k^-\e_j^-&=&0 \\
\e_j^+\e_k^+ + \e_k^+\e_j^+&=&0 \\
 \e_j^-\e_k^+ +\e_k^+\e_j^-&=&\delta_{jk}.
\end{array}
\end{eqnarray}

Conversely, any generator of $\cl_{n,n}$ may be rewritten as a
linear combination involving the basis elements of
$\mbox{End}(\cl_{0,n})$. The remaining linear combinations thus
correspond to
\begin{center}
$\e_j=\e_j^--\e_j^+$ and $\e_{j+n}=\e_j^-+\e_j^+$.
\end{center}

\subsection{Discrete multivector functions}\label{discreteMultivectorSubs}

From now on we will adopt the multivector representation
$x=\sum_{j=1}^n x_j\e_j$ when one refers to the $n-$ tuple
$(x_1,x_2,\ldots, x_n)$ of $\BR^n$ and the displacements $x\pm h
\e_j$ along the $x_j-$axis when one refer to forward/backward shifts
$(x_1,x_2,\ldots,x_j\pm h,\ldots, x_n)$ over the lattice $h\BZ^n$
with mesh width $h>0$. Any $\cl_{0,n}-$valued function $\f(x)$ may
thus be represented as
\begin{eqnarray*}
\f(x)=\sum_{r=0}^n \sum_{|J|=r} f_J(x)\e_J, &~\mbox{with}~&~x \in
h\BZ^n.
\end{eqnarray*}

The algebra containing all the lattice functions $f_J(x)$ will be
denoted by $\mathcal{A}_h$ whereas the linear space containing the
summands of the form $\sum_{|J|=r} f_J(x)\e_J$ will be denoted by
$\Lambda^r \mathcal{A}_h$. Clearly, one has $\Lambda^0
\mathcal{A}_h=\mathcal{A}_h$ and graded direct sum decomposition
 $$
\Lambda^*\mathcal{A}_h=\bigoplus_{r=0}^n \Lambda^r \mathcal{A}_h.
$$

With the aim of filling the lack of commutativity over the lattice
$h\BZ^n$ one will associate to the basis elements $\e_j^-=\e_j
\wedge (\cdot)$ and
 $\e_j^+=\e_j \bullet (\cdot)$ provided
by (\ref{WittBasis}) the following noncommutative actions over
$\Lambda^*\mathcal{A}_h$:
\begin{eqnarray}
\label{TranlationsEpmj}
\begin{array}{ccc}
\e_j^+ \f(x)&=&\f(x-h\e_j)'~\e_j^+ \\ \e_j^-
\f(x)&=&\f(x+h\e_j)'~\e_j^-.
\end{array}
\end{eqnarray}

It is clear from (\ref{TranlationsEpmj}) that the action of each
$\e_j^+\e_j^-$ resp. $\e_j^-\e_j^+$ on $\f(x)$ are commutative.
Also, from (\ref{WittBasisjk}) one can see that the set of operators
$\e_j^+\e_j^-$ resp. $\e_j^-\e_j^+$ mutually commute and satisfy,
for each $j=1,2,\ldots,n$, the set of idempotent relations
$\e_j^+\e_j^-\left(\e_j^+\e_j^-\f(x)\right)=\e_j^+\e_j^-\f(x)$ resp.
$\e_j^-\e_j^+\left(\e_j^-\e_j^+\f(x)\right)=\e_j^-\e_j^+\f(x)$.
 Moreover,
from (\ref{WittBasis}) and (\ref{CliffordBasis}) the mappings
\begin{eqnarray*}
\e_j^+\e_j^-&:&\f(x)\mapsto \e_j \bullet (\e_j \wedge \f(x)) \\
\e_j^-\e_j^+&:&\f(x)\mapsto \e_j \wedge (\e_j \bullet \f(x))
\end{eqnarray*} have the following $\cl_{n,n}$ representation:
\begin{eqnarray*}
\e_j^+\e_j^-=\frac{1}{2}\left(1+\e_{j+n}\e_j\right) & ~~\mbox{and}~~
& \e_j^-\e_j^+=\frac{1}{2}\left(1-\e_{j+n}\e_j\right).
\end{eqnarray*}

Here we would like to observe that from (\ref{CliffordBasis}) the
set of graded commutator mappings on
$\mbox{End}(\Lambda^*\mathcal{A}_h)$, defined as
$$[\e_j^+,\e_j^-]:\f(x) \mapsto
\e_j \bullet (\e_j \wedge \f(x))-\e_j \wedge (\e_j \bullet \f(x)),$$
correspond in $\cl_{n,n}$ to the bivectors $\e_{j+n}\e_j$. Moreover,
they mutually commute and satisfy, for each $j=1,2\ldots,n$, the set
of unitary relations
$$[\e_j^+,\e_j^-]\left([\e_j^+,\e_j^-]\f(x)\right)=\f(x).$$ That
allows us to obtain two projection operators as idempotents of
$\mbox{End}(\Lambda^*\mathcal{A}_h)$.
  To this end, let us introduce the following {\it local} action on
  $h\BZ^n$:
\begin{eqnarray}
\label{staggeredFermion}
\chi_h(x)=\prod_{j=1}^n(-1)^{\frac{x_j}{h}}[\e_j^+,\e_j^-].
\end{eqnarray}

A short computation shows that
$\chi_h(x)^2=\prod_{j=1}^n(-1)^{\frac{2x_j}{h}}[\e_j^+,\e_j^-]^2=1$,
and thus, the unitary relation
$\chi_h(x)\left(\chi_h(x)\f(x)\right)=\f(x)$. Therefore, elements of
the form $\frac{1}{2}\left(1\pm \chi_h(x)\right)\in
\mbox{End}(\Lambda^*\mathcal{A}_h)$ are also idempotent, that is
\begin{eqnarray*}\frac{1}{2}\left(1\pm
\chi_h(x)\right)\left[\frac{1}{2}\left(\f(x)\pm
\chi_h(x)\f(x)\right)\right]
 =\frac{1}{2}\left(\f(x)\pm
\chi_h(x)\f(x)\right).
\end{eqnarray*}

This allows us to introduce the multivector spaces $\Lambda^*_\pm
\mathcal{A}_h$, defined as follows:
\begin{eqnarray*}
\Lambda^*_+ \mathcal{A}_h&=&\left\{ \frac{1}{2}\left(\f(x)+
\chi_h(x)\f(x)\right)~:~\f(x) \in \Lambda^* \mathcal{A}_h \right\}
\\ \Lambda^*_- \mathcal{A}_h&=&
\left\{ \frac{1}{2}\left(\f(x)- \chi_h(x)\f(x)\right)~:~\f(x) \in
\Lambda^* \mathcal{A}_h   \right\}.
\end{eqnarray*}

In concrete, for any lattice function $\f_{+-}(x)$ with membership
in $\Lambda^*_+\mathcal{A}_h \oplus \Lambda^*_-\mathcal{A}_h$ there
exist two multivector functions $\u(x),\v(x) \in
\Lambda^*\mathcal{A}_h$ such that
\begin{eqnarray}
\label{ClnnDirectSum} \f_{+-}(x)&=&\frac{1}{2}\left(\u(x)+
\chi_h(x)\u(x)\right)+\frac{1}{2}\left(\v(x)- \chi_h(x)\v(x)\right).
\end{eqnarray}

The uniqueness of $\u(x)$ and $\v(x)$ is thus assured by the null
relations $\left(1-\chi_h(x)\right)\left(\u(x)+\chi_h(x)\u(x)\right)
=\left(1+\chi_h(x)\right)\left(\v(x)-\chi_h(x)\v(x)\right)=0$.
 Indeed, by letting
act $\frac{1}{2}\left(1+ \chi_h(x)\right)$ and
$\frac{1}{2}\left(1-\chi_h(x)\right)$
 on both sides of (\ref{ClnnDirectSum})
it follows
\begin{eqnarray}
\label{ClnnProjection}
\begin{array}{lll}
\frac{1}{2}\left(\f_{+-}(x)+ \chi_h(x)\f_{+-}(x)\right)=
\frac{1}{2}\left(\u(x)+ \chi_h(x)\u(x)\right) \\ \ \\
\frac{1}{2}\left(\f_{+-}(x)- \chi_h(x)\f_{+-}(x)\right)
=\frac{1}{2}\left(\v(x)- \chi_h(x)\v(x)\right).
\end{array}
\end{eqnarray}

We will use the subscript notations $\f_+(x)$ and $\f_-(x)$ to
denote the multivector functions of the form
$\f_\pm(x)=\frac{1}{2}\left(\f(x) \pm \chi_h(x)\f(x)\right)$. The
bold notations $\f(x),\g(x),\ldots,\u(x)$ and so on will be only
used when we refer to a multivector function belonging to
$\Lambda^*\mathcal{A}_h$.

\begin{remark}
In dimension $n=4$, the direct sum $\Lambda^*_+\mathcal{A}_h \oplus
\Lambda^*_-\mathcal{A}_h$ looks like a {\it Dirac-like spinor}
structure based on the homogeneous representation of the {\it
special unitary group} $SU(2)$ modulo the idempotents
$\frac{1}{2}(1+\chi_h(x))$ and $\frac{1}{2}(1-\chi_h(x))$. Such
ladder structure gives rise to two independent irreducible
representations of the {\it special orthogonal group} $SO(3)$.

When one takes the multivector extension of $\BC^4$ as the
noncommutative ring of quaternions $\mathbb{H}$, from the
isomorphism $\cl_{0,3}\cong \mathbb{H} \oplus \mathbb{H}$ it follows
therefore that $\Lambda^*_+\mathcal{A}_h \oplus
\Lambda^*_-\mathcal{A}_h$ may also be represented, up to the
permutation sign $\displaystyle (-1)^{\sum_{j=1}^n\frac{x_j}{h}}$,
in terms of Dirac matrices $\gamma_j$, with $j=0,1,2,3$
(cf.~\cite[Subsections 2.2 \& 2.3]{Vaz97WittenEq}).

We refer the reader to \cite[Subsection IV.H.]{BorsNielsen00} on
which such discussion involving a $SU(2)\times SU(2)$ representation
for the Lorentz group was taken in the context of Weyl bispinors.
\end{remark}

\subsection{Discrete Dirac operators}

After defining the spaces of multivector functions through the last
subsection, one move now to the construction of finite difference
discretizations for the Dirac operator in intertwining with the
formulations \cite{KK04} and \cite{FKS07}. In order to proceed we
define, for each $j=1,2,\ldots,n$, the forward/backward finite
difference operators $\partial_h^{+j}/\partial_h^{-j}$ by the
coordinate formulae
\begin{eqnarray}
\label{DiffPmj}
\partial_{h}^{+j}\f(x)=\dfrac{\f(x+h\e_j)-\f(x)}{h} &\mbox{and} &
\partial_{h}^{-j}\f(x)=\dfrac{\f(x)-\f(x-h\e_j)}{h}.
\end{eqnarray}

 Take into account the forward/backward finite difference Dirac operators
$D^{\pm}_h=\sum_{j=1}^n\e_j
\partial_{h}^{\pm j}$ already considered in
\cite{FaustinoMonomiality14}, we associate to each lattice function
$\f(x)$ the multivector actions
\begin{center}
$\partial_h^+ \f(x)=D_h^+ \bullet \f(x)$ and $\partial_h^-
\f(x)=D_h^- \wedge \f(x)$.
\end{center}

By means of the Witt basis $\e_j^\pm$ defined in (\ref{WittBasis}),
the actions $\partial_h^+=D_h^+ \bullet (\cdot)$ and
$\partial_h^-=D_h^- \wedge (\cdot)$ on $\Lambda^*\mathcal{A}_h$
resp. $\Lambda^*_+\mathcal{A}_h\oplus \Lambda^*_-\mathcal{A}_h$
correspond to
\begin{eqnarray}
\label{WittDh}
\partial_h^+= \sum_{j=1}^n \e_j^+ \partial_h^{+j} & ~~\mbox{and}~~ &
\partial_h^-= \sum_{j=1}^n \e_j^- \partial_h^{-j}.
\end{eqnarray}

It is clear from (\ref{DiffPmj})
 that $\partial_{h}^{+j}$ and
$\partial_{h}^{-j}$ are interrelated by the shift operators
$S_h^{\pm j}\f(x)=\f(x\pm h\e_j)$, i.e.
\begin{eqnarray}
\label{TranlationsPmj}
\begin{array}{ccc}
S_h^{-j}(\partial_h^{+j}\f(x))=&\partial_h^{+j}(S_h^{-j}\f(x))=&\partial_h^{-j}\f(x)\\
S_h^{+j}(\partial_h^{-j}\f(x))=&\partial_h^{-j}(S_h^{+j}\f(x))=&\partial_h^{+j}\f(x).
\end{array}
\end{eqnarray}

Moreover, for two Clifford-vector-valued functions $\f(x)$ and
$\g(x)$ the action of each $\partial_h^{-j}/\partial_h^{+j}$ on
$\f(x)\g(x)$ gives rise to the set of product rules
\begin{eqnarray}
\label{productRule}
\begin{array}{ccc}
\partial_h^{+j}\left( \f(x)\g(x)  \right)&=&\partial_h^{+
j}\f(x)~\g(x)+\f(x+h\e_j)\partial_h^{+j}\g(x)\\ \ \\
\partial_h^{- j}\left( \f(x)\g(x)  \right)
&=&\partial_h^{- j}\f(x)~\g(x)+\f(x-h\e_j)\partial_h^{- j}\g(x).
\end{array}
\end{eqnarray}

Now let us examine the actions of $\partial_h^+$ and $\partial_h^-$
on Clifford-vector-valued lattice functions. First, recall that the
Gra\ss maniann identities $\e_j^\pm \e_k^\pm+\e_k^\pm \e_j^\pm=0$
provided by (\ref{WittBasisjk}) lead to the nilpotent relations
$$ \partial_h^+\left(\partial_h^+\f(x)\right)=\partial_h^-\left(\partial_h^-\f(x)\right)=0.$$

Based on (\ref{TranlationsEpmj}) , one can also prove the following
lemma, corresponding to generalized Leibniz rules at the level of
discrete multivector calculus.
\begin{lemma}\label{nilpotencyWittDh}
For two lattice functions $\f(x)$ and $\g(x)$ with membership in
$\Lambda^*\mathcal{A}_h$, we have the generalized Leibniz rules
\begin{eqnarray*}
\partial_h^{+}\left(\f(x)\g(x)\right)&=&(\partial_h^{+}\f)(x)\g(x)+\f(x)'(\partial_h^{+}\g)(x)
\\
\partial_h^{-}\left(\f(x)\g(x)\right)&=&(\partial_h^{-}\f)(x)\g(x)+\f(x)'(\partial_h^{-}\g)(x).
\end{eqnarray*}
\end{lemma}

\begin{proof}
Starting from the endomorphism representation provided by
(\ref{WittDh}), the summand(s) splitting
\begin{eqnarray*}
\partial_h^{\pm}(\f(x)\g(x))&=&\sum_{j=1}^n \e_j^{\pm}
\partial_h^{\pm j}\f(x)\g(x)+\sum_{j=1}^n \e_j^{\pm}\f(x\pm h\e_j)(\partial_h^{\pm j}\g)(x)
\end{eqnarray*}
 follows straightforwardly from
 direct application of the product rules (\ref{productRule}) and
from linearity arguments.

The summand(s) $\sum_{j=1}^n \e_j^{\pm} \partial_h^{\pm
j}\f(x)\g(x)$ equals to $(\partial_h^{\pm}\f)(x)\g(x)$ whereas from
noncommutative constraints (\ref{TranlationsEpmj})
\begin{eqnarray*}
\sum_{j=1}^n \e_j^{\pm}\f(x\pm h\e_j)(\partial_h^{\pm j}\g)(x) &=&
\sum_{j=1}^n \f(x)'\e_j^{\pm}(\partial_h^{\pm j}\g)(x) \\
&=& \f(x)'(\partial_h^{\pm}\g)(x).
\end{eqnarray*}
This results into
$\partial_h^{\pm}(\f(x)\g(x))=(\partial_h^{\pm}\f)(x)\g(x)+\f(x)'(\partial_h^{\pm}\g)(x)$,
as desired.
\end{proof}

The above properties altogether show in turn that, for a given
algebra $\mathcal{A}_h$ of real-valued functions over the lattice
$h\BZ^n$ with mesh width $h>0$, the pair
$(\partial_h^-,\Lambda^*\mathcal{A}_h)$ encodes a universal
differential calculus over a hypercubic lattice
(cf.~\cite{DimakisHoissen94,Vaz97,KK04}) whereas $\partial_h^+$
plays the role of the codifferential operator. Indeed, one has the
raising and lowering properties, $\partial_h^{-}: \Lambda^r
\mathcal{A}_h \mapsto \Lambda^{r+1} \mathcal{A}_h$ and
$\partial_h^{+}: \Lambda^r \mathcal{A}_h \mapsto \Lambda^{r-1}
\mathcal{A}_h$, respectively.

Having in mind the Dirac-K\"ahler formalism over differential forms
(cf. \cite[Subsection 4.3.3]{MontvayMunster94}), one can introduce
the finite difference Dirac operator $D_h$ over
$\Lambda_+^*\mathcal{A}_h\oplus \Lambda_-^*\mathcal{A}_h$ as
$D_h=\partial_h^- - \partial_h^{+}$.

It is easy to see from the splittings $\e_j=\e_j^--\e_j^+$ and
$\e_{j+n}=\e_j^-+\e_j^+$ provided by (\ref{WittBasis}) that $D_h$
admits the following $\cl_{n,n}$ based representation
\begin{eqnarray}
\label{DhClnn}
D_h=\sum_{j=1}^n\e_j\frac{\partial_h^{-j}+\partial_h^{+j}}{2}+
\e_{j+n}\frac{\partial_h^{-j}-\partial_h^{+j}}{2}.
\end{eqnarray}

The above formula, corresponding to the decomposition of $D_h$
through $\cl_{n,n}$ into a symmetric plus a skew-symmetric part,
roughly shows that $D_h$ also gives rise to a finite difference
discretization of $\displaystyle D=\sum_{j=1}^n\e_j
\partial_{x_j}$ (cf. \cite{FKS07}) for which the symmetric part corresponds to the central finite
difference Dirac operator $\frac{1}{2}\left(D_h^- + D_h^+\right)$
and the skew-symmetric part equals to $\frac{h}{2}\square_h$, where
$$\square_h=-\sum_{j=1}^n \e_{j+n}
\partial_h^{-j}\partial_h^{+j}$$
denotes the $\cl_{n,0}-$valued extension of the lattice d'Alembert
operators
$-\partial_h^{-j}\partial_h^{+j}=\frac{1}{h}\left(\partial_{h}^{-j}-\partial_{h}^{+j}\right)$
(cf.~\cite[Subsection 1.5.1]{MontvayMunster94}). 

There many basic properties regarding the multivector operators
(\ref{WittDh}) and (\ref{DhClnn}) that can be formulated only in
terms of the graded commuting relations (\ref{WittBasisjk}) and
(\ref{CliffordBasis}), respectively. The next proposition, involving
the factorization of the star Laplacian
\begin{eqnarray}
\label{starLapl}\Delta_h \f(x)=\sum_{j=1}^n
\dfrac{\f(x+h\e_j)+\f(x-h\e_j)-2\f(x)}{h^2}
\end{eqnarray}
will be of special interest on the subsequent section.

\begin{proposition}\label{starLaplFactorization}
For the finite difference discretization of $\Delta_h$ defined by
equation (\ref{starLapl}) it holds
\begin{eqnarray*}
\partial_h^+
\left(\partial_h^-\f(x)\right)+ \partial_h^-
\left(\partial_h^+\f(x)\right)=\Delta_h\f(x)
\\
D_h\left(D_h\f(x)\right)=-\Delta_h\f(x).
\end{eqnarray*}
\end{proposition}

\begin{proof}
Recall that from (\ref{DiffPmj}) one has, for each $j=1,2,\ldots,n$,
the factorization relations
$$ \partial_h^{+j}\left(\partial_h^{-j}\f(x)\right)=\dfrac{\f(x+h\e_j)+\f(x-h\e_j)-2\f(x)}{h^2}.$$

 On the other hand, since from (\ref{TranlationsPmj})
the forward and backward operators $\partial_h^{+j}/\partial_h^{-j}$
mutually commute, it follows from the duality relations $\e_j^+
\e_k^-+\e_k^- \e_j^+=\delta_{jk}$ provided by (\ref{WittBasisjk}),
the set of identities
\begin{eqnarray*}
\Delta_h \f(x)&=&\sum_{j,k=1}^n \left(\e_j^+ \e_k^-+\e_k^-
\e_j^+\right)\partial_h^{+j}(\partial_h^{-k}\f(x))\\
&=&\partial_h^{+}(\partial_h^{-}\f(x))
+\partial_h^{-}(\partial_h^{+}\f(x)).
\end{eqnarray*}

Similarity, the combination of (\ref{CliffordBasis}) with the finite
difference property
\begin{center}
 $\displaystyle
\partial_h^{+j}\partial_h^{-j}=
\left( \frac{\partial_h^{-j}+\partial_h^{+j}}{2}\right)^2-\left(
\frac{\partial_h^{-j}-\partial_h^{+j}}{2}\right)^2$
\end{center}
lead to
\begin{eqnarray*}
-\Delta_h \f(x)&=&\sum_{j=1}^n \e_{j}^2
\left(\frac{\partial_h^{-j}+\partial_h^{+j}}{2}\right)^2\f(x)+\e_{j+n}^2
\left(\frac{\partial_h^{-j}-\partial_h^{+j}}{2}\right)^2\f(x)
\\
&=&-D_h(D_h\f(x)).
\end{eqnarray*}
\end{proof}

\begin{remark}
The statement $D_h(D_h\f(x))=-\Delta_h \f(x)$ also yields as direct
consequence of the combination of the formula $\partial_h^+
\left(\partial_h^-\f(x)\right)+ \partial_h^-
\left(\partial_h^+\f(x)\right)=-D_h\left(D_h\f(x)\right)=\Delta_h\f(x)$
with Lemma \ref{nilpotencyWittDh}.
\end{remark}

\section{Klein-Gordon and Dirac equations on the lattice}\label{KleinGordonDiracSection}

\subsection{The factorization approach}

After the preliminary construction provided in Section
\ref{MultivectorCalculusLattice} one now enters into the heart of
the matter. Our first task consists in the exact formulation of an
exact discretized model for the Dirac equation. Unlike the
Dirac-K\"ahler formulations \cite{BecherJoos82,BorsNielsen00,KK04}
one will incorporate the {\it local} unitary action $\chi_h(x)$ on
$\Lambda_+^* \mathcal{A}_h\oplus\Lambda_-^* \mathcal{A}_h$ to ensure
the factorization of the discretized Klein-Gordon operator
$-\Delta_h+m^2$ provided by the equation (\ref{KleinGordonLattice})
in terms of a discretized Dirac-field operator. Such action is close
to the Dirac-Hestenes spinor field action (cf. \cite[Subsection
3.7]{RodriguesOliveira07}).

To this end, we consider for a multivector function $\f_{+-}(x)$
with membership in $\Lambda_+^* \mathcal{A}_h\oplus\Lambda_-^*
\mathcal{A}_h$, the following discretized Dirac equation on $h\BZ^n$
for a {\it free} particle with mass $m$:
\begin{eqnarray}
\label{DhMass} D_h ~\f_{+-}(x)=m~\chi_h(x)~\f_{+-}(x).
\end{eqnarray}

 From the substitution provided
by (\ref{ClnnDirectSum}) on both sides of (\ref{DhMass}), it can be
easily concluded, based on the idempotency relation
$\chi_h(x)\left(\chi_h(x) \f(x)\right)=\f(x)$, the following
equivalent formulation as a coupled system of equations
\begin{eqnarray}
\label{DhMassCoupled} \left\{\begin{array}{lll}
D_h ~\f_+(x)&=&m~\f_+(x) \\
 D_h~
\f_-(x)&=&~-m~\f_-(x)
\end{array}\right.
\end{eqnarray}
so that the solutions $\f_+(x)$ and $\f_-(x)$ of
(\ref{DhMassCoupled}) are the chiral and the achiral components of
the spinor vector-field $\f_{+-}(x)$.

Considering simply the multivector finite difference operator
$D_h-m\chi_h(x)$, the solutions of (\ref{DhMass}) can be described
in terms of the linear space $\ker\left(D_h-m\chi_h(x)\right)$
whereas the linear spaces $\frac{1}{2}\left(1\pm
\chi_h(x)\right)\left[\ker\left(D_h-m\chi_h(x)\right)\right]$ of
chiral/achiral type contain the solutions of (\ref{DhMassCoupled}).
In addition, the $\cl_{n,n}$-valued representation underlying $D_h$
and $\chi_h(x)$ allows us to derive the subsequent set of results in
a reliable way.
\begin{proposition}\label{DiracLatticeProposition}
When acting on the multivector space $\Lambda_+^*
\mathcal{A}_h\oplus\Lambda_-^* \mathcal{A}_h$, the discretized
Dirac-field operator $D_h-m\chi_h(x)$ satisfies
$$\left(D_h-m\chi_h(x)\right)^2=-\Delta_h+m^2.$$ Moreover, one has
$$\ker\left(D_h-m\chi_h(x)\right)
=\left(D_h-m\chi_h(x)\right)\left[\ker\left(-\Delta_h+m^2\right)\right].$$
\end{proposition}

\begin{proof}
From the direct computation
\begin{eqnarray*}\left(D_h-m\chi_h(x)\right)^2
=D_h^2-mD_h\chi_h(x)-m\chi_h(x)D_h+m^2\chi_h(x)^2
\end{eqnarray*}
 one easily
recognizes the star Laplacian splitting
$\left(D_h\right)^2=-\Delta_h$ (Lemma \ref{starLaplFactorization})
and the unitary relation $\chi_h(x)^2=1$ as well.

On the other hand, from (\ref{CliffordBasis}) we observe that the
set of anti-commuting relations
$\e_j(\e_{j+n}\e_j)=-(\e_{j+n}\e_j)\e_j$ and
$\e_{j+n}(\e_{j+n}\e_j)=-(\e_{j+n}\e_j)\e_{j+n}$ holds for each
$j=1,2,\ldots,n$. This results into the following set of graded
anti-commuting relations
\begin{center}
$\e_j\chi_h(x)=-\chi_h(x)\e_j$ and
$\e_{j+n}\chi_h(x)=-\chi_h(x)\e_{j+n}$,
\end{center}
and moreover, into the basic identity $D_h\chi_h(x)=-\chi_h(x)D_h$.
Whence
$$\left(D_h-m\chi_h(x)\right)^2=-\Delta_h+m^2.$$

For the proof of
$\ker\left(D_h-m\chi_h(x)\right)=\left(D_h-m\chi_h(x)\right)\left[\ker\left(-\Delta_h+m^2\right)\right]$
one recall first that from the factorization
$\left(D_h-m\chi_h(x)\right)^2=-\Delta_h+m^2$, each solution of
(\ref{DhMass}) is also a solution of the discretized Klein-Gordon
equation (\ref{KleinGordonLattice}), and thus, the linear space
$\ker\left(D_h-m\chi_h(x)\right)$ is a subspace of
$\ker\left(-\Delta_h+m^2\right)$.

Therefore, we obtain the inclusion
$$\ker\left(D_h-m\chi_h(x)\right)\subset
\left(D_h-m\chi_h(x)\right)\left[\ker\left(-\Delta_h+m^2\right)\right].$$

Conversely, let $\f_{+-}(x)$ be a multivector function with
membership in
$\left(D_h-m\chi_h(x)\right)\left[\ker\left(-\Delta_h+m^2\right)\right]$,
that is $$\f_{+-}(x)=D_h\g(x)-m\chi_h(x)\g(x),$$ with
$\Delta_h\g(x)=m^2\g(x)$.

By applying $D_h-m\chi_h(x)$ to $\f_{+-}(x)$, it follows then
\begin{eqnarray*}
D_h \f_{+-}(x)-m\chi_h(x)\f_{+-}(x)&=&-\Delta_h\g(x)+m^2\g(x)=0.
\end{eqnarray*}

Thus $\f_{+-}(x)\in \ker\left(D_h-m\chi_h(x)\right)$ and whence
\begin{center}
$\left(D_h-m\chi_h(x)\right)\ker\left(-\Delta_h+m^2\right)
\subset\ker\left(D_h-m\chi_h(x)\right)$.
\end{center}
\end{proof}

\begin{corollary}\label{DiracLatticeCorollary}
Let $\g(x)\in \Lambda^* \mathcal{A}_h$ be a solution of the
discretized Klein-Gordon equation
$$ \Delta_h \g(x)=m^2\g(x).$$
Then, the solutions of the coupled system (\ref{DhMassCoupled}) are
given by
\begin{eqnarray*}
\begin{array}{lll} \f_+(x)
&=&\frac{1}{2}\left(1+\chi_h(x)\right)\left(D_h\g(x)-m\g(x)\right)
\\ \ \\
\f_-(x)
&=&\frac{1}{2}\left(1-\chi_h(x)\right)\left(D_h\g(x)-m\g(x)\right).
\end{array}
\end{eqnarray*}
\end{corollary}

\begin{remark}
The above corollary is also valid for solutions of the discretized
Klein-Gordon equation (\ref{KleinGordonLattice}) with membership in
$\Lambda^*_+ \mathcal{A}_h\oplus\Lambda^*_- \mathcal{A}_h$.
\end{remark}

\subsection{Hypercomplex extension of Chebyshev polynomials}

As an hypercomplex extension for the Chebyshev polynomials of the
first kind labeled by (\ref{ChebyshevPoly}) we can define, for each
$x \in h\BZ^n$, $y,\alpha\in\BR^n$ and $\a\in \cl_{0,n}$, the
Clifford-vector-valued polynomial $\Th^{(\alpha)}(x,y;\a)$ by the
multi-variable formula
\begin{eqnarray}
\label{hypercomplexChebyshev}
\Th^{(\alpha)}(x,y;\a)=
\prod_{j=1}^n{~}_2F_1
\left(-\frac{x_j}{h},\frac{x_j}{h};\frac{1}{2},\frac{1-y_j}{2}+\alpha_j~\chi_h(x)
\right)\a.
\end{eqnarray}

Here and elsewhere ${~}_2F_1$ denotes the hypergeometric series
expansion defined from the Pochhammer symbol $(\lambda)_s
=\frac{\Gamma(\lambda+s)}{\Gamma(\lambda)}$ by the formula
$$
{~}_2F_1(b_1,b_2;c_1;t)=\sum_{s=0}^\infty
\frac{(b_1)_s~(b_2)_s}{(c_1)_s}~\frac{t^s}{s!}.
$$
\begin{remark}
A striking aspect besides this construction is that the term
$\frac{1-y_j}{2}+\alpha_j~\chi_h(x)$ within each ${~}_2F_1$
representation is indeed an element of
$\mbox{End}(\Lambda^*\mathcal{A}_h)$ so that each
${~}_2F_1\left(-\frac{x_j}{h},\frac{x_j}{h};\frac{1}{2},\frac{1-y_j}{2}+\alpha_j~\chi_h(x)
\right)$ shall be understood as an operational action. This means
that the right-hand side of (\ref{hypercomplexChebyshev}) is in fact
a $n-$composition formula.
\end{remark}

Let us now take a close look to the multi-variable representation of
(\ref{hypercomplexChebyshev}) in terms of the classical Chebyshev
polynomials of the first kind $T_k(\lambda)$. For $\alpha={\bf 0}$
(the zero $n-$vector), the symmetric identity ${~}_2F_1\left(
-k,k;\frac{1}{2};\frac{1-\lambda}{2} \right)={~}_2F_1\left(
k,-k;\frac{1}{2};\frac{1-\lambda}{2} \right)$ allows us to extend
the Chebyshev polynomials also for negative integers. That is,
putting $k=\frac{|x_j|}{h}$ on the right-hand side of
(\ref{ChebyshevPoly}) one always have the identity
$T_{\frac{|x_j|}{h}}(y_j)={~}_2F_1\left(-\frac{x_j}{h},\frac{x_j}{h};\frac{1}{2},\frac{1-y_j}{2}
\right)$.

So one has
\begin{eqnarray}
\label{hypercomplexChebyshevMultiVariable} \Th^{({\bf
0})}(x,y;\a)=\frac{\a}{2^n}\prod_{j=1}^n
\left(y_j+\sqrt{y_j^2-1}\right)^{\frac{x_j}{h}}+
\left(y_j-\sqrt{y_j^2-1}\right)^{\frac{x_j}{h}}.
\end{eqnarray}

For $\alpha \neq {\bf 0}$, the idempotency of $\frac{1\pm
\chi_h(x)}{2}$ allows us to represent, for suitable values of $y$,
the hypercomplex polynomials (\ref{hypercomplexChebyshev}) as
elements with membership in $\Lambda_+^*\mathcal{A}_h$ or in
$\Lambda_-^*\mathcal{A}_h$. So one starts with the binomial identity
\begin{eqnarray*}
\left( \frac{1-\lambda}{2}+\mu \frac{1\pm \chi_h(x)}{2} \right)^s
&=&\sum_{r=0}^s \left(\begin{array}{lll}
s \\
r
\end{array}\right)
\left( \frac{1-\lambda}{2}\right)^{s-r} \mu^r \frac{1\pm \chi_h(x)}{2}\\
&=&\frac{1\pm \chi_h(x)}{2}\left( \frac{1-\lambda}{2}+\mu\right)^s.
\end{eqnarray*}

Application of the above formula results, after a straightforwardly
computation based on linearity arguments, into the hypergeometric
identity
\begin{eqnarray*}
{~}_2F_1\left(-k,k;\frac{1}{2},\frac{1-\lambda}{2}+\mu\frac{1\pm
\chi_h(x)}{2}\right)=\frac{1\pm
\chi_h(x)}{2}{~}_2F_1\left(-k,k;\frac{1}{2},\frac{1-\lambda}{2}+\mu\right).
\end{eqnarray*}

Recalling (\ref{hypercomplexChebyshev}) one finds for the
substitutions $\lambda=y_j+2\alpha_j$, $\mu=\pm \alpha_j$, and
$k=\frac{x_j}{h}$, the set of identities
\begin{eqnarray*}
{~}_2F_1\left(-\frac{x_j}{h},\frac{x_j}{h};\frac{1}{2},\frac{1-y_j-2\alpha_j}{2}\pm
2\alpha_j\frac{1\pm \chi_h(x)}{2}
  \right)=\\
  =\frac{1}{2}\left(T_{\frac{|x_j|}{h}}(y_j+ 2\alpha_j)\pm
\chi_h(x)T_{\frac{|x_j|}{h}}(y_j+ 2\alpha_j)\right).
\end{eqnarray*}

After some straightforward manipulations based again on the
idempotency of $\displaystyle \frac{1\pm \chi_h(x)}{2}$, one thus
gets, for $\alpha \neq {\bf 0}$, the set of relations
\begin{eqnarray}
\label{hypercomplexChebyshevProjection}
\begin{array}{lll}
\Th^{(\alpha)}(x,y;\a)=\frac{1}{2}\left(\Th^{({\bf
0})}(x,y+2\alpha;\a) +\chi_h(x)\Th^{({\bf
0})}(x,y+2\alpha;\a)\right) \\
\ \\ \Th^{(-\alpha)}(x,y;\a)=\frac{1}{2}\left(\Th^{({\bf
0})}(x,y+2\alpha;\a)-\chi_h(x)\Th^{({\bf 0})}(x,y+2\alpha;\a)\right)
\end{array}
\end{eqnarray}
that naturally result into the following lemma.
\begin{lemma}
For each $x \in h\BZ^n$, $y,\alpha\in\BR^n$ and $\a\in \cl_{0,n}$ we
then have the following:
\begin{enumerate}
\item The Clifford-vector-valued polynomials of the form
$\Th^{({\bf 0})}(x,y+2\alpha;\a)$ belong to
$\Lambda^*\mathcal{A}_h$. \item For $\alpha \neq {\bf 0}$, the
Clifford-vector-valued polynomials of the form
$\Th^{(\alpha)}(x,y;\a)$ resp. $\Th^{(-\alpha)}(x,y;\a)$ belong to
$\Lambda^*_+\mathcal{A}_h$ resp. $\Lambda^*_-\mathcal{A}_h$.
\end{enumerate}
\end{lemma}

The above description shows that the polynomials
$\Th^{(\alpha)}(x,y;\a)$ and $\Th^{(-\alpha)}(x,y;\a)$ can be
obtained by projecting the multi-variable $\cl_{0,n}-$valued
Chebyshev polynomials $\Th^{({\bf 0})}(x,y+2\alpha;\a)$ on the
multivector spaces $\Lambda_+^*\mathcal{A}_h$ and
$\Lambda_-^*\mathcal{A}_h$, respectively. From this
interrelationship the next proposition is rather obvious.

\begin{proposition}\label{hypercomplexChebyshevRecurrence}
For each $j=1,2,\ldots,n$ the hypercomplex polynomials
$\Th^{(\alpha)}(x,y;\a)$ defined in (\ref{hypercomplexChebyshev})
satisfy the set of three-term recurrence relations
\begin{eqnarray*}
\Th^{(\pm \alpha)}(x+h\e_j,y;\a)+\Th^{(\pm \alpha)}(x-h\e_j,y;\a)
=(2y_j+4\alpha_j)\Th^{(\pm \alpha)}(x,y;\a).
\end{eqnarray*}
\end{proposition}

\begin{proof}
From the set of relations (\ref{hypercomplexChebyshevProjection}),
it is sufficient to show that the three-term recurrence formula
holds for every $\Th^{({\bf 0})}(x+h\e_j,y+2\alpha;\a)$.

First, recall that for each $\gamma\in \cl_{0,n}$ satisfying
$\gamma^2=1$, the auxiliar function of the form
$G_h(t,\lambda;\gamma)=\left(
\lambda+\gamma\sqrt{\lambda^2-1}\right)^{\frac{t}{h}}$ satisfies the
recurrence formulae
\begin{eqnarray*} G_h(t+
h,\lambda;\gamma)&=&\left(\lambda+\gamma\sqrt{\lambda^2-1}\right)G_h(t,\lambda;\gamma)
\\ G_h(t- h,\lambda;\gamma)&=&
\frac{G_h(t,\lambda;\gamma)}{\lambda+\gamma\sqrt{\lambda^2-1}}.
\end{eqnarray*}

This gives rise to
\begin{eqnarray*}
G_h(t+h,\lambda;\gamma)+G_h(t-h,\lambda;\gamma)&=&
\frac{\left(\lambda+\gamma\sqrt{\lambda^2-1}\right)^2+1}{\lambda+\gamma\sqrt{\lambda^2-1}}
G_h(t,\lambda;\gamma)\\
&=&\frac{2\lambda^2+2\gamma\lambda\sqrt{\lambda^2-1}}
{\lambda+\gamma\sqrt{\lambda^2-1}}G_h(t,\lambda;\gamma)\\
&=&2\lambda~G_h(t,\lambda;\gamma).
\end{eqnarray*}

Since the Chebyshev polynomials $T_{\frac{|x_j|}{h}}(y_j+2\alpha_j)$
can be rewritten as
\begin{eqnarray*}
T_{\frac{|x_j|}{h}}(y_j+2\alpha_j)&=&
\frac{1}{2}\left(G_h(x_j,y_j+2\alpha_j;1)+G_h(x_j,y_j+2\alpha_j;-1)
\right),
\end{eqnarray*}
 from the substitutions
$\lambda=y_j+2\alpha_j$ we then have
\begin{eqnarray*}
T_{\frac{|x_j+h|}{h}}(y_j+2\alpha_j)+T_{\frac{|x_j-h|}{h}}(y_j+2\alpha_j)
=(2y_j+4\alpha_j) T_{\frac{|x_j|}{h}}(y_j+2\alpha_j)
\end{eqnarray*}
 so that, for each
$j=1,2,\ldots,n$, the set of three-term recurrence formulae
\begin{eqnarray*}
\Th^{({\bf 0})}(x+h\e_j,y+2\alpha;\a)+\Th^{({\bf
0})}(x-h\e_j,y+2\alpha;\a)=(2y_j+4\alpha_j)~\Th^{({\bf
0})}(x,y+2\alpha;\a).
\end{eqnarray*}
follows naturally in the view of the multi-variable representation
(\ref{hypercomplexChebyshevMultiVariable}).
\end{proof}

Now we are in conditions to write down the solutions of the
discretized Klein-Gordon equation (\ref{KleinGordonLattice}) in
terms of the hypercomplex counterpart of the Chebyshev polynomials
provided by (\ref{hypercomplexChebyshev}) and moreover, the spinor
vector-field components of the Dirac equation (\ref{DhMass}).

 First, recall that based on the coordinate formula
  (\ref{starLapl}) one can rewrite equation
(\ref{KleinGordonLattice}) as
\begin{eqnarray}
\label{KleinGordonLatticeMeanEq}
\sum_{j=1}^n\f(x+h\e_j)+\f(x-h\e_j)=\left((mh)^2+2n\right)\f(x).
\end{eqnarray}

Furthermore, Proposition \ref{hypercomplexChebyshevRecurrence} tells
us that the constraint
$$\sum_{j=1}^n (2y_j+4\alpha_j)=(mh)^2+2n$$ is
a necessary condition for $\Th^{({\bf 0})}(x,y+2\alpha;\a)$, and
moreover, $\Th^{(\pm \alpha)}(x,y;\a)$ being solutions of
(\ref{KleinGordonLatticeMeanEq}).

In the view of Corollary \ref{DiracLatticeCorollary}, if we choose
$y$ and $\alpha$ in such way that $\displaystyle \sum_{j=1}^n
y_j=\frac{(mh)^2}{2}$ and $\displaystyle
\sum_{j=1}^n\alpha_j=\frac{n}{2}$, one can compute the solutions of
(\ref{DhMassCoupled}) from
\begin{center}
$\g(x)=\Th^{(\alpha)}(x,y;\a)+\Th^{(-\alpha)}(x,y;\a)$.
\end{center}

A short computation based on projection arguments and on the
intertwining property $D_h \chi_h(x)=-\chi_h(x)D_h$ even shows that
the solutions $\f_+(x)$ and $\f_-(x)$ of (\ref{DhMassCoupled}) are
thus given by
\begin{eqnarray*}
\begin{array}{lll} \f_+(x)
&=&D_h\Th^{(-\alpha)}(x,y;\a)-m\Th^{(\alpha)}(x,y;\a)
\\ \ \\
\f_-(x) &=&D_h\Th^{(\alpha)}(x,y;\a)-m\Th^{(-\alpha)}(x,y;\a).
\end{array}
\end{eqnarray*}

\section{Some remarks on lattice fermion doubling}\label{FermionDoublingSection}

In relativistic quantum mechanics the Hamiltonian $-\Delta+m^2$
encoded on the time-harmonic Klein-Gordon equation $\Delta \f(x)=m^2
\f(x)$ can be derived by quantizating the paravector representations
$\zeta^+=\xi_0~ -i\xi$ and $\zeta^-=-(\xi_0~+i\xi)$ of the
complexified Clifford algebra $\BC\otimes\cl_{0,n}$ for which
$\xi_0\in\BR$, $i^2=-1$ and $\xi=\sum_{j=1}^n\xi_j \e_j\in
\cl_{0,n}$.

Namely, such quantization scheme can be formulated in terms of the
energy term $E=mc^2$ and the Dirac operator $D=\sum_{j=1}^n\e_j
\partial_{x_j}$ from the prime quantization rules
\begin{eqnarray}
\label{primeQuantization} \xi_0 \mapsto  iE &~~ \mbox{and}~~ & \xi_j
\mapsto i\partial_{x_j}.
\end{eqnarray}

Thus, the substitution $\xi_0^2=\frac{E^2}{c^4}$ (assuming the
energy normalization in terms of the speed of light) gives rise to
the quadratic equations $\zeta^+\zeta^- =\zeta^-
\zeta^+=\sum_{j=1}^n \xi_j^2-m^2$ that in turn yield the
factorization relations, analogous to the ones obtained in
\cite{CFK11}:
\begin{eqnarray*}
(D-im)(D+im)=(D+im)(D-im)=-\Delta+m^2.
\end{eqnarray*}

Therefore, the solutions of the Klein-Gordon equation in {\it
continuum} can be formulated on the momentum space through the
energy-momentum relation $\sum_{j=1}^n \xi_j^2=m^2$.

When one replaces the {\it continuum} Dirac operator $D$ by the
central difference Dirac operator
$\frac{1}{2}\left(D_{h/2}^{-}+D_{h/2}^{+}\right)$ one can mimecking
the above quantization procedure also for the derivation of the
discretized Klein-Gordon equation. That is, based on the coordinate
action of $\frac{1}{2}\left(D_{h/2}^{-}+D_{h/2}^{+}\right)$ on
$\f(x)$ given by
\begin{eqnarray*}
\frac{1}{2}\left(D_{h/2}^{-}\f(x)+D_{h/2}^{+}\f(x)\right)=\sum_{j=1}^n
\e_j
\frac{\f\left(x+\frac{h}{2}\e_j\right)-\f\left(x-\frac{h}{2}\e_j\right)}{h}
\end{eqnarray*}
and from the Taylor series representation $S_{h/2}^{\pm
j}=\exp\left(\pm \frac{h}{2}\partial_{x_j}\right)$ carrying the
shift operators $S_{h/2}^{\pm j}\f(x)=\f\left(x\pm
\frac{h}{2}\e_j\right)$ (cf. \cite[Subsection
2.2]{FaustinoMonomiality14}) it follows that the factorization of
the discretized Klein-Gordon operator $-\Delta_h+m^2$ provided by
the set of relations
\begin{eqnarray*}
-\Delta_h+m^2&=&\left(\frac{1}{2}D_{h/2}^{-}+\frac{1}{2}D_{h/2}^{+}+im\right)
\left(\frac{1}{2}D_{h/2}^{-}+\frac{1}{2}D_{h/2}^{+}-im\right)\\
&=&\left(\frac{1}{2}D_{h/2}^{-}+\frac{1}{2}D_{h/2}^{+}-im\right)
\left(\frac{1}{2}D_{h/2}^{-}+\frac{1}{2}D_{h/2}^{+}+im\right)
\end{eqnarray*}
encodes the quantization correspondence $\displaystyle \sum_{j=1}^n
i\e_j~\frac{2}{h}\sin\left(\frac{h}{2}\xi_j \right)\rightarrow
\frac{1}{2}\left(D_{h/2}^{-}+D_{h/2}^{+}\right)$ and the following
energy-momentum relation on $h\BZ^n$ as well:
$$\sum_{j=1}^n \frac{4}{h^2}\sin^2\left(\frac{h\xi_j}{2}
\right)=m^2.$$

On the flavor of finite difference potentials (cf.
\cite{GH01,CKKS14}) such kind of quantization on the lattice that
results into the aforementioned energy-momentum relation was already
considered, in a {\it hidden} way, when the authors obtained
integral representations involving Green's-type functions based on
the fact that the restriction of the continuous Fourier transform to
$Q_h=\left[ -\frac{\pi}{h},\frac{\pi}{h}\right]^n$ gives an inverse
for the discrete Fourier transform over $h\BZ^n$ (see also
\cite[Section 1.5]{MontvayMunster94}, \cite[Section II.]{daVeiga02}
and \cite[Subsection 4.2]{FaustinoMonomiality14} for further
analogies and comparisons).

The drawback besides this quantization scheme is two-fold: for the
factorization of the discretized Klein-Gordon equation
(\ref{KleinGordonLattice}) one needs to consider, in addition the
lattice $\frac{h}{2}\BZ^n$ with mesh width $\frac{h}{2}>0$ and the
left-hand side of the above summand has $2n+1$ zeros inside the {\it
Brillouin zone} $Q_h=\left[ -\frac{\pi}{h},\frac{\pi}{h}\right]^n$,
and so, such discretization yields spectrum degeneracy for
$\frac{1}{2}\left(D_{h/2}^{-}+D_{h/2}^{+}\right)$ in the massless
limit $m\rightarrow 0$, as already discussed in Subsection
\ref{StateArtSubsection}.

When one formulates the discretized Dirac equation
$D_h\f(x)=m\chi_h(x)\f(x)$ from the finite difference discretization
$D_h=\frac{1}{2}\left( D_h^-+D_h^+\right)+\frac{h}{2}\square_h$
provided by (\ref{DhClnn}), one assures that its momentum
representation provided by the prime quantization rules
$\xi_j\mapsto i\partial_{x_j}$ and $-i\partial_{\xi_j}\mapsto x_j$
over $\cl_{n,n}$:
\begin{eqnarray}
\label{DhClnnMomentum}
 \sum_{j=1}^n
 i\e_j~\frac{1}{h}\sin(h\xi_j)\widehat{\f}(\xi)+\sin^2\left(\frac{h\xi_j}{2}\right)
 \e_{j+n}\widehat{\f}(\xi)=m\prod_{j=1}^n
 \e_{j+n}\e_j\cosh\left(\frac{\pi}{h}\partial_{\xi_j}\right)\widehat{\f}(\xi)
\end{eqnarray}
behaves like $\frac{1}{2}\left(D_{h}^{-}+D_{h}^{+}\right)$ in the
neighborhood of $\xi_j=0$ and $\xi_j=\pm \frac{\pi}{2h}$ and to
avoid gaps in a neighborhood of $\xi_j=\pm \frac{\pi}{h}$. We remark
that this behavior agrees with the results obtained by A.~Bori\c ci
\cite{Borici08} and M.~Creutz \cite{Creutz08}.

Contrary to Kanamori-Kawamoto's approach (cf.~\cite[Subsection
6]{KK04}), the mass term $m$ was treated as a {\it local} spinor
field potential $m~\chi_h(x)$ that produces, as it was explained in
Subsection \ref{discreteMultivectorSubs}, a commutative action over
the lattice $h\BZ^n$.
 The connection with staggered fermions of Kogut-Susskind type
was recognized from the shift operators
$$\cosh\left(\frac{\pi}{h}\partial_{\xi_j}\right):
\widehat{\f}(\xi)\mapsto \frac{\widehat{\f}\left(\xi+
\frac{\pi}{h}\e_j\right)+\widehat{\f}\left(\xi-
\frac{\pi}{h}\e_j\right)}{2},$$ as representations of the
permutation term $(-1)^{\frac{x_j}{h}} =\cos\left( \frac{\pi
x_j}{h}\right)$ on the momentum space (cf. \cite[Subsection
4.3]{MontvayMunster94}). This rotation symmetry action on $h\BZ^n$
is thus essential to assure that to the corners of the {\it
Brillouin zone} $Q_h=\left[ -\frac{\pi}{h},\frac{\pi}{h}\right]^n$
(that is, the doublers) goes to a single point in the limit
$h\rightarrow 0$.

At the level of the phase space coordinates $(x,\xi)\in h\BZ^n
\times [-\frac{\pi}{h},\frac{\pi}{h}]$, one can further introduce
the {\it local} transformation actions
\begin{eqnarray*}
\widehat{\f}(\xi)\mapsto
\Th^{\left(\frac{1}{2}\e\right)}(x,y;\a)\widehat{\f}(\xi) &
\mbox{and} & \widehat{\f}(\xi)\mapsto
\Th^{\left(-\frac{1}{2}\e\right)}(x,y;\a)\widehat{\f}(\xi)
\end{eqnarray*}
 underlying the multivector representations
\begin{center}
$\displaystyle y=\sum_{j=1}^n 2\sin^2\left(
\frac{h\xi_j}{2}\right)\e_j$ and $\displaystyle
\frac{1}{2}\e=\sum_{j=1}^n \frac{1}{2}\e_j$
\end{center}
 of the
$n-$tuples $\left(2\sin^2\left(
\frac{h\xi_1}{2}\right),2\sin^2\left(
\frac{h\xi_2}{2}\right),\ldots,2\sin^2\left(
\frac{h\xi_n}{2}\right)\right)$ and
$\left(\frac{1}{2},\frac{1}{2},\ldots,\frac{1}{2}\right)$,
respectively.

Essentially, these local transformations provide families of
solutions for the discretized Klein-Gordon equation as well as they
assure the existence of two conserved spinorial flows of chiral and
achiral type, respectively, for the momentum representation
(\ref{DhClnnMomentum}). From Corollary \ref{DiracLatticeCorollary},
one can also see that the action of the Dirac-field operator
$D_h-m\chi_h(x)$ on such transformations produce solutions for the
Dirac equation (\ref{DhMass}) over the phase space $h\BZ^n\times
\left[ -\frac{\pi}{h},\frac{\pi}{h}\right]^n$. Moreover, the
solutions of (\ref{DhMassCoupled}) with membership in
$\Lambda_+^*\mathcal{A}_h\oplus \Lambda_-^*\mathcal{A}_h$, given by
\begin{eqnarray*}
\begin{array}{lll} \widehat{\f}_+(x,\xi)
&=&D_h\left(\Th^{\left(-\frac{1}{2}\e\right)}(x,y;\a)\widehat{\f}(\xi)\right)-
m\Th^{\left(\frac{1}{2}\e\right)}(x,y;\a)\widehat{\f}(\xi)
\\ \ \\
\widehat{\f}_-(x,\xi)
&=&D_h\left(\Th^{\left(\frac{1}{2}\e\right)}(x,y;\a)\widehat{\f}(\xi)\right)
-m\Th^{\left(-\frac{1}{2}\e\right)}(x,y;\a)\widehat{\f}(\xi).
\end{array}
\end{eqnarray*}
provide the required spinor vector-field components of
(\ref{DhMass}).

\subsection*{Acknowledgment} I would like to thank to Jayme Vaz Jr.
(IMECC-UNICAMP, Brazil) for driving my attention at an earlier stage
to Rabin's approach \cite{Rabin82} and for his research paper
\cite{Vaz97WittenEq}, on which the construction of spinor-like
spaces, analogue to the ones constructed in Section
\ref{MultivectorCalculusLattice}, was considered.

This paper was further developed in depth after a guest visit to
School of Mathematics, University of Leeds (UK), from April 22 till
April 30 2014. I would also like to thank to Vladimir~V.~Kisil for
the fruitful discussions around this topic during these days and for
his surge of interest on the feasibility of this approach.

Last but not least a special acknowledgment to Artan Bori\c ci
(University of Tirana, Albania) and to Paulo A. F. da Veiga
(ICMC-USP, Brazil) for the references
\cite{Borici08,Creutz08,daVeiga02}, and to the anonymous referees
for the careful readership of the paper.



\end{document}